\documentclass{llncs}
\usepackage{xspace,amsmath,amssymb,stmaryrd,enumerate,color,ulem}
\usepackage[all]{xy}

\newcommand{\red}[1]{\textcolor{red}{#1}}
\newcommand{\blue}[1]{\textcolor{blue}{#1}}
\newcommand{\U}{\mathcal{U}}
\renewcommand{\O}{\mathcal{O}}
\newcommand{\F}{\mathcal{F}}
\newcommand{\R}{\mathcal{R}}
\renewcommand{\P}{\mathcal{P}}
\newcommand{\sem}[1]{\llbracket #1 \rrbracket}
\newcommand{\define}{\mbox{\it define}}
\newcommand{\constrain}{\mbox{\it constrain}}
\newcommand{\Cond}{\mbox{\it cond}}
\newcommand{\CondPr}{\mbox{\it cond'}}
\newcommand{\Def}{\mbox{\it def}}
\newcommand{\DefPr}{\mbox{\it def'}}
\newcommand{\DefDPr}{\mbox{\it def''}}
\newcommand{\Con}{\mbox{\it con}}
\newcommand{\true}{\mbox{\it true}}
\newcommand{\false}{\mbox{\it false}}

\newcommand{\AX}{\mbox{\bf AX}}
\newcommand{\EG}{\mbox{\bf EG}}
\newcommand{\AG}{\mbox{\bf AG}}
\newcommand{\E}{\mbox{\bf E}}
\newcommand{\A}{\mbox{\bf A}}
\newcommand{\Until}{\mbox{\bf U}}
\newcommand{\G}{\mbox{\bf G}}
\newcommand{\X}{\mbox{\bf X}}
\newcommand{\Paths}{\mbox{\it Paths}}

\begin{document}

\title{Layered Fixed Point Logic}
\author{Piotr Filipiuk %\inst{1}
   \and Flemming Nielson %\inst{1}
   \and Hanne Riis Nielson %\inst{1}
   }
\institute{DTU Informatics, Richard Petersens Plads,Technical University of Denmark, DK-2800 Kongens Lyngby, Denmark \\
           \email{\{pifi,nielson,riis\}@imm.dtu.dk}}

\maketitle

\begin{abstract}
  We present a logic for the specification of static analysis problems
  that goes beyond the logics traditionally used. Its most prominent
  feature is the direct support for both inductive computations of
  behaviors as well as co-inductive specifications of properties. Two
  main theoretical contributions are a Moore Family result and a
  parametrized worst case time complexity result. We show that the
  logic and the associated solver can be used for rapid prototyping
  and illustrate a wide variety of applications within Static
  Analysis, Constraint Satisfaction Problems and Model Checking. In
  all cases the complexity result specializes to the worst case time
  complexity of the classical methods.
\end{abstract}

\section{Introduction}
Static analysis \cite{bib:hecht1977,bib:ppa} is a successful approach
to the validation of properties of programming languages. It can be
seen as a two-phase process where we first transform the analysis
problem into a set of constraints that, in the second phase, is solved
to produce the analysis result of interest. The constraints may be
expressed in a language tailored to the problem at hand, or they may
be expressed in a general purpose constraint language such as Datalog
\cite{bib:datalog1,bib:datalog2} or ALFP \cite{bib:ssforalfp}.

Model checking \cite{bib:clarkemc,bib:pmc} is an automatic technique
for verifying hardware and more recently software
systems. Specifications are expressed in modal logic, whereas the
system is modeled as a transition system or a Kripke structure. Given
a system description the model checking algorithm either proves that
the system satisfies the property, or reports a counterexample that
violates it.

Constraint Satisfaction Problems (CSPs) \cite{bib:Mackworth77} are the
subject of intense research in both artificial intelligence and
operations research. They consist of variables with constraints on
them, and many real-world problems can be described as CSPs. A major
challenge in constraint programming is to develop efficient generic
approaches to solve instances of the CSP.

% Modal $\mu$-calculus \cite{bib:kozen83,bib:clarkemc} is extensively used
% in various areas of computer science such as e.g~computer-aided
% verification. The interest stems from the fact that many temporal and
% program logics can be encoded into the $\mu$-calculus. Its defining
% feature is the addition of least and greatest fixpoint operators to
% modal logic; thus it achieves a great increase in expressive power,
% but at the same time an equally great increase in difficulty of
% understanding.

In this paper we present a logic for specification of analysis
problems that goes beyond the logics traditionally used. Its most
prominent feature is the direct support for both inductive
computations of behaviors as well as co-inductive specifications of
properties. At the same time the approach taken falls within the
Abstract Interpretation \cite{bib:cousot79,bib:cousot77} framework,
thus there always is a unique best solution to the analysis problem
considered. We show that the logic and the associated solver can be
used for rapid prototyping and illustrate a wide variety of
applications within Static Analysis, Constraint Satisfaction Problems
and Model Checking.

One can notice a resemblance of the logic to modal $\mu$-calculus
\cite{bib:kozen83,bib:clarkemc}, which is extensively used in various
areas of computer science such as e.g~computer-aided verification. Its
defining feature is the addition of least and greatest fixpoint
operators to modal logic; thus it achieves a great increase in
expressive power, but at the same time an equally great increase in
difficulty of understanding.

The paper is organized as follows. In Section \ref{sec:logic} we
define the syntax and semantics of LFP. In Section
\ref{sec:optimal-solutions} we establish a Moore Family result and
estimate the worst case time complexity. In Section \ref{sec:sa} we show
an application of LFP to Static Analysis. We continue in Section
\ref{sec:csp} with an application to the Constraint Satisfaction
Problem. An application to Model Checking in presented in Section
\ref{sec:mc}. We conclude in Section \ref{sec:conclusions}.

\section{Syntax and Semantics}
\label{sec:logic}
In this section, we introduce Layered Fixed Point Logic (abbreviated
LFP). The LFP formulae are made up of layers. Each layer can either be
a {\it define} formula which corresponds to the inductive definition,
or a {\it constrain} formula corresponding to the co-inductive
specification. The following definition introduces the syntax of LFP.

\begin{definition}
  Given a fixed countable set $\mathcal{X}$ of variables, a non-empty
  universe $\mathcal{U}$, a finite set of function symbols $\F$, and a
  finite alphabet $\mathcal{R}$ of predicate symbols, we define the
  set of LFP formulae, $cls$, together with clauses, $cl$,
  conditions, $\Cond$, constrains, $\Con$, definitions, $\Def$, and terms
  $u$ by the grammar:
  \begin{center}
    \begin{tabular}{ l c l }
      $u$ & ::= & $x \mid f(\vec u)$\\
      $\Cond$ & ::= & $ R(\vec x) \mid \neg R(\vec x) \mid \Cond_1
      \wedge \Cond_2 \mid \Cond_1 \vee \Cond_2$ \\ 
      & $ \mid $ & $ \exists x: \Cond \mid
      \forall x: \Cond \mid \true \mid \false$ \\
      $\Def$ & ::= & $ \Cond \Rightarrow R(\vec u) \mid \forall x: \Def
      \mid \Def_1 \wedge \Def_2$ \\
      $\Con$ & ::= & $ R(\vec u) \Rightarrow \Cond \mid \forall x: \Con
      \mid \Con_1 \wedge \Con_2$ \\
      $cl_i$ & ::= & $ \define(\Def) \mid \constrain(\Con) $ \\
      $cls$ & ::= & $ cl_1,\ldots,cl_s$
    \end{tabular}
  \end{center}
  Here $x\in\mathcal{X}$, $R \in \mathcal{R}$, $f
  \in \F$ and $1 \leq i \leq s$. We say that $s$ is the order of the
  LFP formula $cl_1,\ldots,cl_s$.
\end{definition}

We allow to write $R(\vec u)$ for $\true \Rightarrow R(\vec u)$, $\neg R(\vec
u)$ for $R(\vec u) \Rightarrow \false$ and we abbreviate zero-arity
functions $f()$ as $f \in \mathcal{U}$. Occurrences of $R(\vec x)$ and
$\neg R(\vec x)$ in conditions are called positive and negative
queries, respectively. Occurrences of $R(\vec u)$ on the right hand
side of the implication in define formulas are called defined
occurrences. Occurrences of $R(\vec u)$ on the left hand side of the
implication in constrain formulas are called constrained
occurrences. Defined and constrained occurrences are jointly called
assertions.

In order to ensure desirable theoretical and pragmatic properties in
the presence of negation, we impose a notion of
\textit{stratification} similar to the one in Datalog
\cite{bib:datalog1,bib:datalog2}. Intuitively, stratification ensures
that a negative query is not performed until the predicate has been
fully asserted (defined or constrained). This is important for
ensuring that once a condition evaluates to true it will continue to
be true even after further assertions of predicates.
\begin{definition}\label{def:stratification}
The formula $cl_1,\ldots, cl_s$ is
stratified if for all $i = 1,\ldots, s$ the following properties hold:
\begin{itemize}
\item Relations asserted in $cl_i$ must not be asserted in
  $cl_{i+1},\ldots, cl_s$
\item Relations positively used in $cl_i$ must not be asserted in
  $cl_{i+1},\ldots, cl_s$
\item Relations negatively used in $cl_i$ must not be asserted in
  $cl_i,\ldots, cl_s$
\end{itemize}
The function $rank : \mathcal{R} \rightarrow \{0,\ldots,s\}$ is then uniquely
defined as
$$
rank(R) = \max(\{0\} \cup \{i \mid R \text{ is asserted in }cl_i\})
$$
\end{definition}

\begin{example}
Using the notion of stratification we can define equality $eq$ and
non-equality $neq$ predicates as follows
$$
\define(\forall x: \true \Rightarrow eq(x,x)),\define(\forall
x: \forall y: \neg eq(x,y) \Rightarrow neq(x,y))
$$
According to Definition \ref{def:stratification} the formula is
stratified, since predicate $eq$ is negatively used only in the layer
above the one that defines it.
\end{example}

To specify the semantics of LFP we introduce the interpretations
$\varrho$, $\zeta$ and $\varsigma$ of predicate symbols, function
symbols and variables, respectively. Formally we have
$$
\begin{array}{rl}
  \varrho: & \prod_{k} \mathcal{R}_{/k} \rightarrow \P(\U^k) \\
  \zeta: & \prod_{k} \F_{/k} \rightarrow \U^k
  \rightarrow \U \\
  \varsigma: & \mathcal{X} \rightarrow \mathcal{U} \\
\end{array}
$$
In the above $\mathcal{R}_{/k}$ stands for a set of predicate symbols
of arity $k$, then $\mathcal{R}$ is a disjoint union of
$\mathcal{R}_{/k}$, hence $\mathcal{R}=\ \biguplus_{k}
\mathcal{R}_{/k}$.  Similarity $\F_{/k}$ is a set of function symbols
of arity $k$ and $\F=\ \biguplus_{k} \F_{/k}$. The interpretation of
variables is given by $\sem{x}(\zeta, \varsigma)=\varsigma(x)$, where
$\varsigma(x)$ is the element from $\mathcal{U}$ bound to $x\in{\cal
  X}$. Furthermore, the interpretation of function terms is defined as
$\sem{f(\vec u)}(\zeta,\varsigma)=\sem{f}(\zeta,[\,])(\sem{\vec
  u}(\zeta,\varsigma))$. It is generalized to sequences $\vec u$ of
terms in a point-wise manner by taking $\sem{a}(\zeta, \varsigma)=a$
for all $a\in {\cal U}$, and $\sem{(u_1,\ldots, u_k)}(\zeta,
\varsigma)=(\sem{u_1}(\zeta, \varsigma), \ldots, \sem{u_k}(\zeta,
\varsigma))$.

The satisfaction relations for conditions $\Cond$, definitions $\Def$ and
constrains $\Con$ are specified by:
\[
(\varrho, \varsigma) \models \Cond,\quad (\varrho, \zeta,
\varsigma) \models \Def\quad \mathrm{and}\ (\varrho, \zeta, \varsigma)
\models \Con
\]
The formal definition is given in Table \ref{LFPSemantics}; here
$\varsigma[x\mapsto a]$ stands for the mapping that is as
$\varsigma$ except that $x$ is mapped to $a$.

\begin{table}
\caption{Semantics of LFP}
$$
\begin{array}{lllll}
  (\varrho, \varsigma) & \models & R(\vec x) &
  \underline{\texttt{iff}} & \sem{\vec x}([\,], \varsigma) \in \varrho(R)
  \\
  (\varrho, \varsigma) & \models & \neg R(\vec x) & \underline{\texttt{iff}} & \sem{\vec x}([\,], \varsigma) \notin \varrho(R) \\
  (\varrho, \varsigma) & \models & \Cond_1 \wedge \Cond_2 & \underline{\texttt{iff}} & (\varrho, \varsigma) \models \Cond_1 \text{ and } (\varrho, \varsigma) \models \Cond_2
  \\
  (\varrho, \varsigma) & \models & \Cond_1 \vee \Cond_2 & \underline{\texttt{iff}} & (\varrho, \varsigma) \models \Cond_1 \text{ or } (\varrho, \varsigma) \models \Cond_2 
  \\
  (\varrho, \varsigma) & \models & \exists x:\Cond & \underline{\texttt{iff}} & (\varrho, \varsigma[x \mapsto a]) \models \Cond \text{ for some }a \in \cal{U}
  \\
  (\varrho, \varsigma) & \models & \forall x:\Cond &
  \underline{\texttt{iff}}
  & (\varrho, \varsigma[x \mapsto a]) \models \Cond \text{ for all }a
  \in \cal{U} \\
  (\varrho, \varsigma) & \models & \true &
  \underline{\texttt{iff}} & {\sf always} \\
  (\varrho, \varsigma) & \models & \false &
  \underline{\texttt{iff}} & {\sf never} \\ \\
  (\varrho, \zeta, \varsigma) & \models & R(\vec u) &
  \underline{\texttt{iff}} & \sem{\vec u}(\zeta, \varsigma) \in \varrho(R)  \\
  (\varrho, \zeta, \varsigma) & \models & \Def_1 \wedge \Def_2 &
  \underline{\texttt{iff}} & (\varrho, \zeta, \varsigma) \models \Def_1
  \text{ and } (\varrho, \zeta, \varsigma) \models \Def_2 
  \\
  (\varrho, \zeta, \varsigma) & \models & \Cond \Rightarrow R(\vec u) &
  \underline{\texttt{iff}} & (\varrho, \zeta, \varsigma) \models R(\vec u) \text{ whenever } (\varrho, \varsigma) \models \Cond 
  \\
  (\varrho, \zeta, \varsigma) & \models & \forall x:\Def &
  \underline{\texttt{iff}} & (\varrho, \zeta, \varsigma[x \mapsto a]) \models \Def \text{ for all }a \in \cal{U} 
  \\ \\
  (\varrho, \zeta, \varsigma) & \models & R(\vec u) & \underline{\texttt{iff}} & \sem{\vec u}(\zeta, \varsigma) \in \varrho(R) 
  \\
  (\varrho, \zeta, \varsigma) & \models & \Con_1 \wedge \Con_2 & \underline{\texttt{iff}} & (\varrho, \zeta, \varsigma) \models \Con_1 \text{ and } (\varrho, \zeta, \varsigma) \models \Con_2 
  \\
  (\varrho, \zeta, \varsigma) & \models & R(\vec u) \Rightarrow \Cond &
  \underline{\texttt{iff}} & (\varrho, \varsigma) \models \Cond \text{
    whenever } (\varrho, \zeta, \varsigma) \models R(\vec u) 
  \\
  (\varrho, \zeta, \varsigma) & \models & \forall x:\Con &
  \underline{\texttt{iff}} & (\varrho, \zeta, \varsigma[x \mapsto a]) \models
  \Con \text{ for all }a \in \cal{U} \\ \\

  (\varrho, \zeta, \varsigma) & \models & cl_1 , \ldots, cl_s & \underline{\texttt{iff}} & (\varrho, \zeta, \varsigma) \models
  cl_i \mbox{ for all }1 \leq i \leq s 
\end{array}
$$
\label{LFPSemantics}
\end{table}

\section{Optimal Solutions}
\label{sec:optimal-solutions}
\paragraph{Moore Family.}
First we establish a Moore family result for LFP, which guarantees
that there always is a unique best solution for LFP formulae.

\begin{definition}
A Moore family is a subset $Y$ of a complete lattice $L=(L,\sqsubseteq)$
that is closed under greatest lower bounds: $\forall Y' \subseteq Y:
\bigsqcap Y' \in Y$.
\end{definition}

It follows that a Moore family always contains a least element,
$\bigsqcap Y$, and a greatest element, $\bigsqcap \emptyset$, which
equals the greatest element, $\top$, from $L$; in particular, a Moore
family is never empty. The property is also called the model intersection
property, since whenever we take a {\it meet} of a number of models we still
get a model.

Let $\Delta = \{ \varrho \mid \varrho: \prod_{k} \mathcal{R}_{/k}
\rightarrow \P(\U^k) \}$ denote the set of interpretations $\varrho$
of predicate symbols in $\mathcal{R}$ over $\mathcal{U}$. We define a
lexicographical ordering $\sqsubseteq$ defined by $\varrho_1
\sqsubseteq \varrho_2$ if and only if there is some $0 \leq j \leq s$
, where $s$ is the order of the formula, such that the following
properties hold:
\begin{enumerate}[(a)]
\item $\varrho_1(R)=\varrho_2(R)$ for all $R \in \mathcal{R}$ with
  $rank(R) < j$,\label{itm:rank-less}
\item $\varrho_1(R) \subseteq \varrho_2(R)$ for all $R \in \mathcal{R}$ with
  $rank(R) = j$ and either $j=0$ or $R$ is a {\it defined} relation,\label{itm:rank-eq-def}
\item $\varrho_1(R) \supseteq \varrho_2(R)$ for all $R \in \mathcal{R}$ with
  $rank(R) = j$ and $R$ is a {\it constrained} relation,\label{itm:rank-eq-con}
\item either $j=s$ or $\varrho_1(R) \neq \varrho_2(R)$ for some
  relation $R \in \mathcal{R}$ with $rank(R) = j$.\label{itm:rank-s}
\end{enumerate}

\begin{lemma}\label{lemma:partial-order}
$\sqsubseteq$ defines a partial order.
\end{lemma}
\begin{proof}
See Appendix \ref{proof:lemma:partial-order}.\qed
%\TODO{Prove the lemma.}
\end{proof}

\begin{lemma}\label{lemma:complete-lattice}
$(\Delta, \sqsubseteq)$ is a complete lattice with the greatest lower
bound given by
$$
(\bigsqcap M)(R) = \left\{
\begin{array}{ll} 
  \bigcap \{ \varrho(R) \mid \varrho \in M_j \} & \mbox{if
    $rank(R)=j$ and} \\
  & \mbox{either $j=0$ or $R$ is {\it defined} in $cl_j$.} \\
  \bigcup \{ \varrho(R) \mid \varrho \in M_j \} & \mbox{if
    $rank(R)=j$ and} \\
  & \mbox{$R$ is {\it constrained} in $cl_j$.}
\end{array}\right.
$$
where
$$
M_j = \{ \varrho \in M \mid \forall R': rank(R') < j \Rightarrow
(\bigsqcap M)(R')=\varrho(R') \}
$$
\end{lemma}
\begin{proof}
See Appendix \ref{proof:lemma:complete-lattice}.\qed
\end{proof}

Note that $\bigsqcap M$ is well defined by induction on $j$ observing that
$M_0=M$ and $M_j \subseteq M_{j-1}$.

\begin{proposition}\label{prop:moore-family}
  Assume $cls$ is a stratified LFP formula, $\varsigma_0$ and
  $\zeta_0$ are interpretations of the free variables and function
  symbols in $cls$, respectively. Furthermore, $\varrho_0$ is an
  interpretation of all relations of rank 0. Then $\{ \varrho \mid
  (\varrho, \zeta_0, \varsigma_0) \models cls \wedge \forall R:
  rank(R) = 0 \Rightarrow \varrho(R) \supseteq \varrho_0(R) \}$ is a
  Moore family.
\end{proposition}
\begin{proof}
See Appendix \ref{proof:prop:moore-family}.\qed
\end{proof}

The result ensures that the approach falls within the framework of
Abstract Interpretation \cite{bib:cousot77,bib:cousot79}; hence we can
be sure that there always is a single best solution for the analysis
problem under consideration, namely the one defined in Proposition
\ref{prop:moore-family}.

%\section{Complexity}
%\label{sec:complexity}
\paragraph{Complexity.}
The least model for LFP formulae guaranteed by Proposition
\ref{prop:moore-family} can be computed efficiently as summarized in
the following result.
\begin{proposition}\label{prop:complexity}
  For a finite universe $\U$, the best solution $\varrho$ such that
  $\varrho_0 \sqsubseteq \varrho$ of a LFP formula $cl_1, \ldots,
  cl_s$ (w.r.t. an interpretation of the constant symbols) can be
  computed in time
\[
\mathcal{O}(|\varrho_0| + \sum_{1\leq i \leq s} |cl_i||\mathcal{U}|^{k_i})
\]
where $k_i$ is the maximal nesting depth of quantifiers in the $cl_i$
and $|\varrho_0|$ is the sum of cardinalities of predicates
$\varrho_0(R)$ of rank $0$. We also assume unit time hash table
operations (as in \cite{bib:complex}).
\end{proposition}
\begin{proof}
See Appendix \ref{proof:prop:complexity}.\qed
\end{proof}

For {\it define} clauses a straightforward method that achieves the
above complexity proceeds by instantiating all variables occurring in
the input formula in all possible ways. The resulting formula has no
free variables thus it can be solved by classical solvers for
alternation-free Boolean equation systems \cite{bib:hornsat} in linear
time.

In case of {\it constrain} clauses we first dualize the problem by
transforming the co-inductive specification into the inductive
one. The transformation increases the size of the input formula by a
constant factor. Thereafter, we proceed in the same way as for the
define clauses.

In addition we need to take into account the number of known facts,
which equals to the cardinality of all predicates of rank $0$. As a
result we get the complexity from Proposition \ref{prop:complexity}.

\paragraph{The solver.}
We developed a state-of-the-art solver for LFP, which is implemented
in continuation passing style using Haskell. The solver computes the
least model guaranteed by Proposition \ref{prop:moore-family} and has
a worst case time complexity as given by Proposition
\ref{prop:complexity}. For many clauses it exhibits a running time
substantially lower than the worst case time complexity. Indeed,
\cite{bib:complex} gives a formula estimating the less than worst case
time complexity on a given clause.

The solver deals with stratification by computing the relations in
increasing order on their rank and therefore the negations present no
obstacles. The relations are represented as Ordered Binary Decision
Diagrams (OBDDs), which were originally used in hardware
verification. OBDDs can efficiently store a large number of states
that share many commonalities \cite{bib:obdds,bib:bryant86}, and have
already been used in a number of program analyses proving to be very
efficient. The algorithm is an extension of the symbolic algorithm
presented in \cite{bib:piotr-ss} and is based on the top-down solving
approach of Le Charlier and van Hentenryck \cite{bib:LeCharlier92}. 

The solver automatically translates LFP formulae into highly efficient
OBDD implementations. Since the OBDDs represent sets of tuples, the
solver operates on entire relations at a time, rather than individual
tuples. The cost of the OBDD operations depends on the size of the
OBDD and not the number of tuples in the relation; hence dense
relations can be computed efficiently as long as their encoded
representations are compact.

% \noindent The translation is defined as follows

% $$
% \begin{array}{lll}
% f (\forall x: con) & = & \forall x: f(con) \\
% f (R(\vec u) \Rightarrow pre) & = & let\ pre' = \neg
% pre[R^\complement(\vec u) / \neg R(\vec u)]\ in \\
% & & let\ def = \forall \vec u : \neg R^\complement(\vec u) \Rightarrow
% R(\vec u)\ in \\
% & & (pre' \Rightarrow R^\complement(\vec u)) \wedge def \\
% \end{array}
% $$

\section{Application to Data Flow Analysis}
\label{sec:sa}
Datalog has already been used for program analysis in compilers
\cite{bib:pabdd,bib:reps93,bib:ullman89}. In this section we present
how the LFP logic can be used to specify analyses that are instances
of Bit-Vector Frameworks, which are a special case of the Monotone
Frameworks \cite{bib:ppa,bib:mf}.

A Monotone Framework consists of (a) a property space that usually is
a complete lattice $L$ satisfying the Ascending Chain Condition, and
(b) transfer functions, i.e.~monotone functions from $L$ to
$L$. The property space is used to represent the data flow
information, whereas transfer functions capture the behavior of
actions. In the Bit-Vector Framework, the property space is a power
set of some finite set and all transfer functions are of the form
$f_{n}(x)=(x \setminus kill_n) \cup gen_n$.

Throughout the section we assume that a program is represented as a
control flow graph \cite{bib:kildall,bib:ppa}, which is a directed
graph with one entry node (having no incoming edges) and one exit node
(having no outgoing edges), called extremal nodes. The remaining nodes
represent statements and have transfer functions associated with them.

\paragraph{Backward may analyses.}
Let us first consider backward may analyses expressed as an instance
of the Monotone Frameworks. In the analyses, we require the least sets
that solve the equations and we are able to detect properties
satisfied by at least one path leading to the given node. The analyses
use the reversed edges in the flow graph; hence the data flow
information is propagated {\it against} the flow of the program
starting at the exit node. The data flow equations are defined as
follows
$$
\begin{array}{l}
A(n)=
\left\{
\begin{array}{ll}
\iota & \text{if } n=n_{exit} \\
\bigcup \{ f_n(A(n') \mid (n,n') \in E \} & \text{otherwise} 
\end{array}\right.
\end{array}
$$
where $A(n)$ represents data flow information at the entry to the node
$n$, $E$ is a set of edges in the control flow graph, and $\iota$ is
the initial analysis information. The first case in the above
equation, initializes the exit node with the initial analysis
information, whereas the second one joins the data flow information
from different paths (using the revered flow). We use $\bigcup$ since
we want be able detect properties satisfied by at least one path
leading to the given node.

The LFP specification for backward may analyses consists of two
conjuncts corresponding to two cases in the data flow
equations. Since in case of may analyses we aim at computing the least
solution, the specification is defined in terms of a {\it define}
clause. The formula is obtained as
$$
\begin{array}{l}
\define
\left(
\begin{array}{c}
\forall x: \iota(x) \Rightarrow A(n_{exit},x) \\
\bigwedge_{(s,t)\in E} \forall x: (A(t,x) \wedge \neg kill_s(x)) \vee gen_s(x) \Rightarrow A(s,x)
\end{array}\right)
\end{array}
$$
The first conjunct initializes the exit node with initial analysis
information, denoted by the predicate $\iota$. The second one
propagates data flow information agains the edges in the control flow
graph, i.e.~whenever we have an edge $(s,t)$ in the control flow
graph, we propagate data flow information from $t$ to $s$, by applying
the corresponding transfer function.

Notice that there is no explicit formula for joining analysis
information from different paths, as it is the case in the data flow
equations, but rather it is done implicitly. Suppose there are two
distinct edges $(s,p)$ and $(s,q)$ in the flow graph, then we get
$$
\begin{array}{l}
\forall x: \underbrace{(A(p,x) \wedge \neg kill_{s}(x)) \vee gen_{s}(x)}_{\Cond_p(x)}
\Rightarrow A(s,x)\\
\forall x: \underbrace{(A(q,x) \wedge \neg kill_{s}(x)) \vee gen_{s}(x)}_{\Cond_q(x)}
\Rightarrow A(s,x)
\end{array}
$$
which is equivalent to
$$
\forall x: \Cond_p(x) \vee \Cond_q(x) \Rightarrow A(s,x)
$$

\paragraph{Forward must analyses.}
Let us now consider the general pattern for defining forward must
analyses. Here we require the largest sets that solve the equations
and we are able to detect properties satisfied by all paths leading to
a given node. The analyses propagate the data flow information along
the edges of the flow graph starting at the entry node. The data flow
equations are defined as follows
$$
\begin{array}{l}
A(n)=
\left\{
\begin{array}{ll}
\iota & \text{if } n=n_{entry} \\
\bigcap \{ f_n(A(n')) \mid (n',n) \in E \} & \text{otherwise} 
\end{array}\right.
\end{array}
$$
where $A(n)$ represents analysis information at the exit from the node
$n$. Since we require the greatest solution, the greatest lower bound
$\bigcap$ is used to combine information from different paths.

The corresponding LFP specification is obtained as follows
$$
\begin{array}{l}
\constrain
\left(
\begin{array}{c}
\forall x: A(n_{entry},x) \Rightarrow \iota(x) \\
\bigwedge_{(s,t)\in E} \forall x: A(t,x) \Rightarrow (A(s,x) \wedge \neg kill_t(x)) \vee gen_t(x)
\end{array}\right)
\end{array}
$$
Since we aim at computing the greatest solution, the analysis is given
by means of {\it constrain} clause. The first conjunct initializes the
entry node with the initial analysis information, whereas the second
one propagates the information along the edges in the control flow
graph, i.e.~whenever we have an edge $(s,t)$ in the control flow
graph, we propagate data flow information from $s$ to $t$, by applying
the corresponding transfer function.

The general patterns for defining forward may and backward must
analyses follow similar pattern. In case of forward may analyses the
data flow information is propagated along the edges of the flow graph
and since we aim at computing the least solution, the analyses are
given by means of {\it define} clauses. Backward must analyses, on the
other hand, use reversed edges in the flow graph and are specified
using {\it constrain} clauses.

In order to compute the least solution of the data flow equations, one
can use a general iterative algorithm for Monotone Frameworks. The
worst case complexity of the algorithm is $\O(|E|h)$, where $|E|$ is
the number of edges in the control flow graph, and $h$ is the height
of the underlying lattice \cite{bib:ppa}. For Bit-Vector Frameworks
the lattice is a powerset of a finite set $\U$; hence $h$ is
$\O(|\U|)$. This gives the complexity $\O(|E||\U|)$.

According to Proposition \ref{prop:complexity} the worst case time
complexity of the LFP specification is $\O(|\varrho_0|+\sum_{1\leq i
  \leq |E|}|\U||cl_i|)$. Since the size of the clause $cl_i$ is
constant and the sum of cardinalities of predicates of rank $0$ is
$\O(|N|)$ we get $\O(|N| + |E||\U|)$. Provided that $|E|>|N|$ we
achieve $\O(|E||\U|)$ i.e.~the same worst case complexity as the
standard iterative algorithm.

It is common in the compiler optimization that various analyses are
preformed at the same time. Since LFP logic has direct support for
both least fixed points and greatest fixed points, we can perform both
may and must analyses at the same time by splitting the analyses into
separate layers.

\section{Application to Constraint Satisfaction}
\label{sec:csp}
Arc consistency is a basic technique for solving Constraint
Satisfaction Problems (CSP) and has various applications within
e.g. Artificial Intelligence. Formally a CSP
\cite{bib:Mackworth77,bib:YuanlinY01} problem can be defined as
follows.
\begin{definition}
  A Constraint Satisfaction Problem $(N, D, C)$ consists of a finite
  set of variables $N = \{ x_1,\ldots, x_n \}$, a set of domains $D =
  \{ D_1,\ldots,D_n \}$, where $x_i$ ranges over $D_i$, and a set of
  constraints $C \subseteq \{ c_{ij} \mid i,j \in N \}$, where each
  constraint $c_{ij}$ is a binary relation between variables $x_i$ and
  $x_j$.
\end{definition}

For simplicity we consider binary constraints only. Furthermore, we
can represent a CSP problem as a directed graph in the following way.

\begin{definition}
  A constraint graph of a CSP problem $(N,D,C)$ is a directed graph
  $G=(V,E)$ where $V=N$ and $E=\{ (x_i,x_j) \mid c_{ij} \in C \}$.
\end{definition}

Thus vertices of the graph correspond to the variables and an edge in
the graph between nodes $x_i$ and $x_j$ corresponds to the constraint
$c_{ij} \in C$.

% \begin{example}
%   Consider a CSP problem $(N,D,C)$ where $N=\{x_1,x_2\}$, $D_1=D_2=\{
%   1,\ldots,5 \}$ and $C=\{ x_1+x_2=7, x_1\ mod\ 2 = 0 \}$. We can
%   represent the above CSP problem as a constraint graph depicted in
%   Figure \ref{fig:csp}.
% \end{example}

The arc consistency problem is formally stated in the following
definition.

\begin{definition}\label{def:csp}
  Given a CSP $(N, D, C)$, an arc $(x_i, x_j)$ of its constraint graph
  is arc consistent if and only if $\forall x \in D_i$, there exists
  $y \in D_j$ such that $c_{ij}(x,y)$ holds, as well as $\forall y \in
  D_j$, there exists $x \in D_i$ such that $c_{ij}(x,y)$ holds. A CSP
  $(N,D,C)$ is arc consistent if and only if each arc in its
  constraint graph is arc consistent.
\end{definition}

The basic and widely used arc consistency algorithm is the AC-3
algorithm proposed in 1977 by Mackworth \cite{bib:Mackworth77}. The
complexity of the algorithm is $O(ed^3)$, where $e$ is the number of
constraints and $d$ the size of the largest domain. The algorithm is
used in many constrains solvers due to its simplicity and fairly good
efficiency \cite{bib:wallace93}.

Now we show the LFP specification of the arc consistency problem. A
domain of a variable $x_i$ is represented as a unary relation $D_i$, and
for each constraint $c_{ij} \in C$ we have a binary relation $C_{ij}
\subseteq D_i \times D_j$. Then we obtain
$$
\begin{array}{l}
\constrain
\left( \bigwedge_{c_{ij} \in C}
\begin{array}{l}
  (\forall x: D_i(x) \Rightarrow \exists y: D_j(y) \wedge
  C_{ij}(x,y)) \wedge \\
  (\forall y: D_j(y) \Rightarrow \exists x: D_i(x) \wedge
  C_{ij}(x,y)) \\
\end{array}
\right)
\end{array}
$$
which exactly captures the conditions from Definition \ref{def:csp}.

According to the Proposition \ref{prop:complexity} the above
specification gives rise to the worst case complexity
$\mathcal{O}(ed^2)$. The original AC-3 algorithm was optimized in
\cite{bib:YuanlinY01} where it was shown that it achieves the worst
case optimal time complexity of $\mathcal{O}(ed^2)$. Hence LFP
specification is as efficient as the improved version of the AC-3
algorithm.

\begin{example}
  As an example let us consider the following problem. Assume we have
  two processes $P_1$ and $P_2$ that need to be finished before 8 time
  units have elapsed. The process $P_1$ is required to run for 3 or 4
  time units, the process $P_2$ is required to run for precisely 2
  time units, and $P_2$ should start at the exact moment when $P_1$
  finishes.

  \begin{figure}
  \begin{center}
    \hfill \xymatrix{ *+[o][F-]{s_1} \ar@(l,u)[]^>{c_{11}}
      \ar[rr]^{c_{12}} & & *+[o][F-]{s_2} \ar@(r,u)[]_>{c_{22}} }
    \hfill\mbox{}
  \end{center}
  \caption{Arc consistency.}
  \label{fig:csp}
  \end{figure}
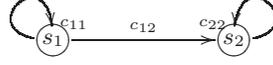

  The problem can be defined as an instance of CSP $(N,D,C)$ where
  $N=\{ s_1, s_2 \}$ denoting the starting times of the corresponding
  process. Since both processes need to be completed before 8 time
  units have elapsed we have $D_1=D_2=\{ 0, \ldots, 8 \}$. Moreover,
  we have the following constrains $C= \{ c_{12} = (3 \leq s_2 -s_1
  \leq 4), c_{11} = (0 \leq s_1 \leq 4), c_{22} = (0 \leq s_2 \leq 6)
  \}$. We can represent the above CSP problem as a constraint graph
  depicted in Figure \ref{fig:csp}. Furthermore it can be specified as
  the following LFP formulae
$$
\begin{array}{l}
\define
\left(
\begin{array}{l}
\bigwedge_{0 \leq x \leq 4} C_1(x) \wedge \bigwedge_{0 \leq y \leq 6}
C_2(y) \wedge \bigwedge_{3 \leq z \leq 4} C_{12}(z)
\end{array}
\right),\\
\constrain
\left(
\begin{array}{l}
  (\forall x: D_1(x) \Rightarrow \exists y: D_2(y) \wedge
  C_{12}(y-x)) \wedge \\
  (\forall y: D_2(y) \Rightarrow \exists x: D_1(x) \wedge
  C_{12}(y-x))
\end{array}
\right)
\end{array}
$$
where we write $y-x$ for a function $\mbox{\it f}_{sub}(y,x)$.
\end{example}

\section{Application to Model Checking}
\label{sec:mc}
This section is concerned with the application of the LFP logic to the
model checking problem \cite{bib:pmc}. In particular we show how LFP
can be used to specify a prototype model checker for a special purpose
modal logic of interest. Here we illustrate the approach on the
familiar case of Computation Tree Logic (CTL)
\cite{bib:ctl}. Throughout this section, we assume that $TS$ is finite
and has no terminal states.

CTL distinguishes between state formulae and path formulae. CTL state
formulae over the set $AP$ of atomic propositions are formed according
to the following grammar
$$
\Phi ::= true \mid a \mid \Phi_1 \wedge \Phi_2 \mid \neg \Phi \mid
\E \varphi \mid \A \varphi
$$
where $a \in AP$ and $\varphi$ is a path formula. CTL path formulae are
formed according to the following grammar
$$
\varphi ::= \X \Phi \mid \Phi_1 \Until \Phi_2 \mid \G \Phi
$$
where $\Phi$, $\Phi_1$ and $\Phi_2$ are state formulae. The
satisfaction relation $\models$ is defined for state formula by
$$
\begin{array}{lcl}
s \models \textbf{true} & \underline{\texttt{iff}} & true \\

s \models a & \underline{\texttt{iff}} & a \in L(s) \\

s \models \neg \Phi & \underline{\texttt{iff}} & \mbox{not } s \models
\Phi \\

s \models \Phi_1 \wedge \Phi_2 & \underline{\texttt{iff}} & s \models
\Phi_1 \mbox{ and } s \models \Phi_2 \\

s \models \E \varphi & \underline{\texttt{iff}} & \pi \models \varphi \mbox{
for some } \pi \in \Paths(s) \\

s \models \A \varphi & \underline{\texttt{iff}} & \pi \models \varphi \mbox{
for all } \pi \in \Paths(s) \\
\end{array}
$$
where $Paths(s)$ denote the set of maximal path fragments $\pi$
starting in $s$. The satisfaction relation $\models$ for path formulae
is defined by
$$
\begin{array}{lcl}
\pi \models \X \Phi & \underline{\texttt{iff}} & \pi[1] \models \Phi
\\

\pi \models \Phi_1 \Until \Phi_2 & \underline{\texttt{iff}} & \exists j
\geq 0: (\pi[j] \models \Phi_2 \wedge (\forall 0 \leq k < j: \pi[k]
\models \Phi_1)) \\

\pi \models \G \Phi & \underline{\texttt{iff}} & \forall j \geq 0:
\pi[j] \models \Phi
\end{array}
$$
where for path $\pi=s_0 s_1\ldots$ and an integer $i\geq0$, $\pi[i]$
denotes the $(i+1)$th state of $\pi$, i.e.~$\pi[i]=s_i$.

The CTL model checking amounts to a recursive computation of the set
$Sat(\Phi)$ of all states satisfying $\Phi$, which is sometimes
referred to as {\it global} model checking. The algorithm boils down
to a bottom-up traversal of the abstract syntax tree of the CTL
formula $\Phi$. The nodes of the abstract syntax tree correspond to
the sub-formulae of $\Phi$, and leaves are either a constant $\true$
or an atomic proposition $a \in AP$.

\begin{table}
\caption{LFP specification of satisfaction sets}
$$
\begin{array}{l}
\define(\forall s: Sat_{\true}(s))\\
\define(\forall s: L_a(s) \Rightarrow Sat_a(s)) \\
\define(\forall s: Sat_{\Phi_1}(s) \wedge Sat_{\Phi_2}(s) \Rightarrow
Sat_{\Phi_1 \wedge \Phi_2}(s)) \\
\define(\forall s: \neg Sat_{\Phi}(s) \Rightarrow Sat_{\neg
  \Phi}(s)) \\ \\

\define(\forall s: (\exists s':T(s,s')\wedge Sat_{\Phi}(s')) \Rightarrow
Sat_{\textbf{EX}\Phi}(s)) \\ \\

\define(\forall s: (\forall s': \neg T(s,s') \vee Sat_{\Phi}(s')) \Rightarrow
Sat_{\textbf{AX}\Phi}(s)) \\ \\

\define
\left(
\begin{array}{l}
  (\forall s: Sat_{\Phi_2}(s) \Rightarrow
  Sat_{\textbf{E}[\Phi_1\textbf{U}\Phi_2]}(s)) \wedge \\
  (\forall s: Sat_{\Phi_1}(s) \wedge (\exists s': T(s,s') \wedge
  Sat_{\textbf{E}[\Phi_1\textbf{U}\Phi_2]}(s')) \Rightarrow
  Sat_{\textbf{E}[\Phi_1\textbf{U}\Phi_2]}(s))
\end{array}
\right)
\\ \\

\define
\left(
\begin{array}{l}
(\forall s: Sat_{\Phi_2}(s) \Rightarrow
Sat_{\textbf{A}[\Phi_1\textbf{U}\Phi_2]}(s)) \wedge \\
(\forall s: Sat_{\Phi_1}(s) \wedge (\forall s': \neg T(s,s') \vee 
Sat_{\textbf{A}[\Phi_1\textbf{U}\Phi_2]}(s')) \Rightarrow
Sat_{\textbf{A}[\Phi_1\textbf{U}\Phi_2]}(s))
\end{array}
\right)
\\ \\

\constrain
\left(
\begin{array}{l}
(\forall s: Sat_{\textbf{EG}\Phi}(s) \Rightarrow Sat_{\Phi}(s)) \wedge \\
(\forall s: Sat_{\textbf{EG}\Phi}(s) \Rightarrow (\exists
 s': T(s,s') \wedge Sat_{\textbf{EG}\Phi}(s')))
\end{array}
\right)\\ \\

\constrain
\left(
\begin{array}{l}
(\forall s: Sat_{\textbf{AG}\Phi}(s) \Rightarrow Sat_{\Phi}(s)) \wedge \\
(\forall s: Sat_{\textbf{AG}\Phi}(s) \Rightarrow (\forall
 s': \neg T(s,s') \vee Sat_{\textbf{AG}\Phi}(s')))
\end{array}
\right)
\end{array}
$$
\label{tab:sat-sets-lfp}
\end{table}

Now let us consider the LFP specification, where for each formula
$\Phi$ we define a relation $Sat_{\Phi} \subseteq S$ characterizing
states where $\Phi$ hold. The specification is defined in Table
\ref{tab:sat-sets-lfp}. The clause for $\true$ is straightforward and
says that $\true$ holds in all states. The clause for an atomic
proposition $a$ expresses that a state satisfies $a$ whenever it is in
$L_a$, where we assume that we have a predicate $L_a \subseteq S$ for
each $a \in AP$. The clause for $\Phi_1 \wedge \Phi_2$ captures that a
state satisfies $\Phi_1 \wedge \Phi_2$ whenever it satisfies both
$\Phi_1$ and $\Phi_2$. Similarly a state satisfies $\neg \Phi$ if it
does not satisfy $\Phi$.
The formula for ${\bf EX}\Phi$ captures that a state $s$ satisfies
${\bf EX}\Phi$, if there is a transition to state $s'$ such that $s'$
satisfies $\Phi$. The formula for ${\bf AX}\Phi$ expresses that a
state $s$ satisfies $\AX\Phi$ if for all states $s'$: either there is
no transition from $s$ to $s'$, or otherwise $s'$ satisfies $\Phi$.
The formula for $\textbf{E}[\Phi_1\textbf{U}\Phi_2]$ captures two
possibilities. If a state satisfies $\Phi_2$ then it also satisfies
$\textbf{E}[\Phi_1\textbf{U}\Phi_2]$. Alternatively if the state $s$
satisfies $\Phi_1$ and there is a transition to a state satisfying
$\textbf{E}[\Phi_1\textbf{U}\Phi_2]$ then $s$ also satisfies
$\textbf{E}[\Phi_1\textbf{U}\Phi_2]$.
The formula $\textbf{A}[\Phi_1\textbf{U}\Phi_2]$ also captures two
cases. If a state satisfies $\Phi_2$ then it also satisfies
$\textbf{A}[\Phi_1\textbf{U}\Phi_2]$. Alternatively state $s$
satisfies $\textbf{A}[\Phi_1\textbf{U}\Phi_2]$ if it satisfies
$\Phi_1$ and for all states $s'$ either there is no transition from
$s$ to $s'$ or $\textbf{A}[\Phi_1\textbf{U}\Phi_2]$ is valid in $s'$.
Let us now consider the formula for $\textbf{EG}\Phi$. Since the set
of states satisfying $\textbf{EG}\Phi$ is defined as a largest set
satisfying the semantics of $\EG\Phi$, the property is defined by
means of constrain clause. The first conjunct expresses that whenever
a state satisfies $\textbf{EG}\Phi$ it also satisfies $\Phi$. The
second conjunct says that if a state satisfies $\textbf{EG}\Phi$ then
there exists a transition to a state $s'$ such that $s'$ satisfies
$\textbf{EG}\Phi$.
Finally let us consider the formula for $\AG\Phi$, which is also
defined in terms of constrain clause and distinguishes between two
cases. In the first one whenever a state satisfies $\AG\Phi$, it also
satisfies $\Phi$. Alternatively, if a state $s$ satisfies $\AG\Phi$
then for all states $s'$: either there is no transition from $s$ to
$s'$ or otherwise $s'$ satisfies $\AG\Phi$.

The generation of clauses for $Sat_{\Phi}$ is performed in the
postorder traversal over $\Phi$; hence the clauses defining
sub-formulas of $\Phi$ are defined in the lower layers. It is
important to note that the specification in Table
\ref{tab:sat-sets-lfp} is both correct and precise.  It follows that
an implementation of the given specification of CTL by means of the
LFP solver constitutes a model checker for CTL.

We may estimate the worst case time complexity of model checking
performed using LFP. Consider a CTL formula $\Phi$ of size $|\Phi|$;
it is immediate that the LFP clause has size $\O(|\Phi|)$, and the
nesting depth is at most 2. According to Proposition
\ref{prop:complexity} the worst case time complexity of the LFP
specification is $\O(|S| + |S|^2|\Phi|)$, where $|S|$ is the number of
states in the transition system. Using a more refined reasoning than
that of Proposition \ref{prop:complexity} we obtain $\O(|S| +
|T||\Phi|)$, where $|T|$ is the number of transitions in the transition
system. It is due to the fact that the "double quantifications'' over
states in Table \ref{tab:sat-sets-lfp} really correspond to traversing
all possible transitions rather than all pairs of states. Thus our LFP
model checking algorithm has the same worst case complexity as
classical model checking algorithms \cite{bib:pmc}.

\begin{example}
  As an example let us consider the Bakery mutual exclusion algorithm
  \cite{bib:bakery}. Although the algorithm is designed for an
  arbitrary number of processes, we consider the simpler setting with
  two processes. Let $P_1$ and $P_2$ be the two processes, and $x_1$
  and $x_2$ be two shared variables both initialized to $0$.  We can
  represent the algorithm as an interleaving of two program graphs
  \cite{bib:pmc}, which are directed graphs where actions label the
  edges rather than the nodes. The algorithm is as follows

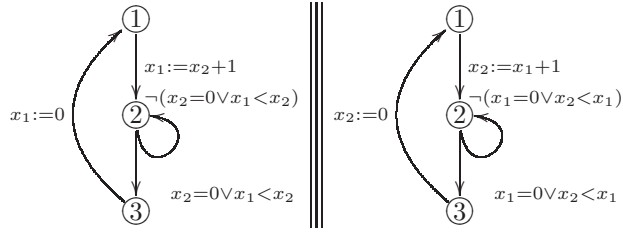
\[
\begin{array}{l|||l}
\xymatrix{
*+[o][F-]{1} \ar[d]^{x_1:=x_2+1} \\
*+[o][F-]{2} \ar[d]^>{\quad x_2=0\vee x_1<x_2} 
\ar@(d,r)[]_>{\qquad\qquad\neg(x_2=0\vee x_1<x_2)}\\
*+[o][F-]{3} \ar @/^2pc/[uu]^{x_1:=0} }
&
\xymatrix{
*+[o][F-]{1} \ar[d]^{x_2:=x_1+1} \\
*+[o][F-]{2} \ar[d]^>{\quad x_1=0\vee x_2<x_1} 
\ar@(d,r)[]_>{\qquad\qquad\neg(x_1=0\vee x_2<x_1)}\\
*+[o][F-]{3} \ar @/^2pc/[uu]^{x_2:=0} }
\end{array}
\]

The variables $x_1$ and $x_2$ are used to resolve the conflict when
both processes want to enter the critical section. When $x_i$ is equal
to zero, the process $P_i$ is not in the critical section and does not
attempt to enter it --- the other one can safely proceed to the
critical section. Otherwise, if both shared variables are non-zero,
the process with smaller ``ticket'' (i.e. value of the corresponding
variable) can enter the critical section. This reasoning is captured
by the conditions of busy-waiting loops. When a process wants to enter
the critical section, it simply takes the next ``ticket'' hence giving
priority to the other process.

From the algorithm above, we can obtain a program graph corresponding
to the interleaving of the two processes, which is depicted in Figure
\ref{fig:interleaved-bakery}.
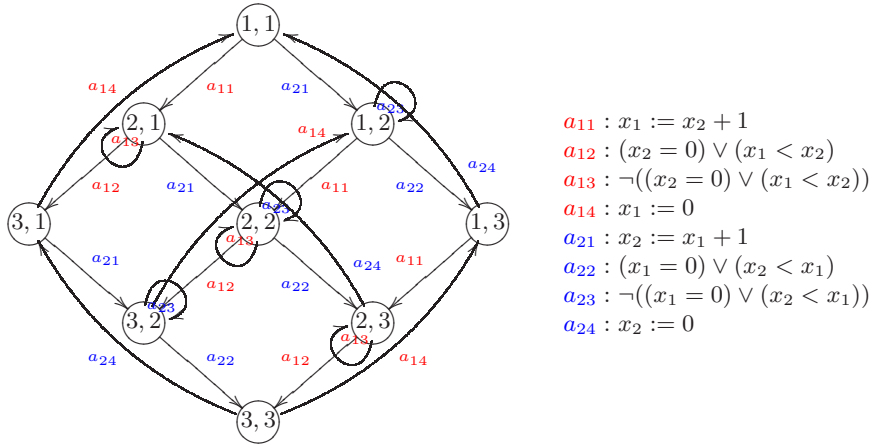
\begin{figure}
\begin{minipage}{0.6\textwidth}
$
\xymatrix{
&& *+[o][F-]{1,1}
\ar[dl]^{\red{a_{11}}}
\ar[dr]_{\blue{a_{21}}} 
&& 
\\
& *+[o][F-]{2,1}
\ar[dl]^{\red{a_{12}}}
\ar[dr]_{\blue{a_{21}}}
\ar@(d,l)[]_>{\red{a_{13}}}
&& *+[o][F-]{1,2}
\ar[dl]^{\red{a_{11}}}
\ar[dr]_{\blue{a_{22}}}
\ar@(u,r)[]_>{\blue{a_{23}}}
\\
*+[o][F-]{3,1}
\ar[dr]^{\blue{a_{21}}}
\ar@/^1pc/[uurr]^{\red{a_{14}}}
&& *+[o][F-]{2,2}
\ar[dl]^{\red{{a_{12}}}}
\ar[dr]_{\blue{{a_{22}}}}
\ar@(d,l)[]_>{\red{a_{13}}}
\ar@(u,r)[]_>{\blue{a_{23}}}
&& *+[o][F-]{1,3}
\ar[dl]_{\red{a_{11}}}
\ar@/_1pc/[uull]_<<<<<{\blue{a_{24}}}
\\
& *+[o][F-]{3,2}
\ar[dr]^{\blue{{a_{22}}}}
\ar@/^1pc/[uurr]^>>>{\red{a_{14}}}
\ar@(u,r)[]_>{\blue{a_{23}}}
&& *+[o][F-]{2,3}
\ar[dl]_{\red{{a_{12}}}}
\ar@/_1pc/[uull]_<<<<<{\blue{a_{24}}}
\ar@(d,l)[]_>{\red{a_{13}}}
\\
&& *+[o][F-]{3,3}
\ar@/^1pc/[uull]^{\blue{a_{24}}}
\ar@/_1pc/[uurr]_{\red{a_{14}}}
}
$
\end{minipage}
\begin{minipage}{0.4\textwidth}
$
\begin{array}{ll}
\red{a_{11}}: & x_1:=x_2+1\\
\red{a_{12}}: & (x_2=0) \vee (x_1<x_2)\\
\red{a_{13}}: & \neg((x_2=0) \vee (x_1<x_2))\\
\red{a_{14}}: & x_1:=0\\
\blue{a_{21}}: & x_2:=x_1+1\\
\blue{a_{22}}: & (x_1=0) \vee (x_2<x_1)\\
\blue{a_{23}}: & \neg((x_1=0) \vee (x_2<x_1))\\
\blue{a_{24}}: & x_2:=0
\end{array}
$
\end{minipage}
\caption{Interleaved program graph.}
\label{fig:interleaved-bakery}
\end{figure}

The CTL formulation of the mutual exclusion property is $AG \neg
(crit_1 \wedge crit_2)$, which states that along all paths globally it
is never the case that $crit_1$ and $crit_2$ hold at the same
time. 
% Notice that in the above formula, we have a least fixed point
% formula, namely $\neg (crit_1 \wedge crit_2)$, nested in the $AG$
% modality that is a greatest fixed point.

As already mentioned, in order to specify the problem we proceed
bottom up by specifying formulae for the sub problems. After a bit of
simplification we obtain the following LFP clauses
$$
\begin{array}{l}
\define(\forall s: L_{crit_1}(s) \wedge L_{crit_2}(s) \Rightarrow
Sat_{crit}(s)), \\
\constrain
\left(
\begin{array}{l}
(\forall s: Sat_{AG( \neg crit)}(s) \Rightarrow \neg
Sat_{crit}(s)) \wedge \\
(\forall s: Sat_{AG( \neg crit)}(s) \Rightarrow (\forall s': \neg T(s,s')
\vee Sat_{AG( \neg crit)}(s')))
\end{array}
\right)
\end{array}
$$
where relation $L_{crit_1}$ (respectively $L_{crit_1}$) characterizes
states in the interleaved program graph that correspond to process
$P_1$ (respectively $P_2$) being in the critical section. Furthermore,
the $AG$ modality is defined by means of a constrain clause. The first
conjunct expresses that whenever a state satisfies a mutual exclusion
property $AG( \neg crit)$ it does not satisfy $crit$. The second one
states that if a state satisfies a mutual exclusion property then all
successors do as well, i.e.~for an arbitrary state, it is either not a
successor or else satisfies the mutual exclusion property.
\end{example}

\section{Conclusions}
\label{sec:conclusions}
In the paper we introduced the Layered Fixed Point Logic, which is a
suitable formalism for the specification of analysis problems. Its
most prominent feature is the direct support for both inductive as
well as co-inductive specifications of properties.

We established a Moore Family result that guarantees that there always
is a best solution for the LFP formulae. More generally this ensures
that the approach taken falls within the general Abstract
Interpretation framework. Other theoretical contribution is the
parametrized worst case time complexity result, which provide a simple
characterization of the running time of the LFP programs.

We developed a state-of-the-art solving algorithm for LFP, which is a
continuation passing style algorithm based on OBDD representations of
relations. The solver achieves the best known theoretical complexity
bounds, and for many clauses exhibit a running time substantially
lower than the worst case time complexity.

We showed that the logic and the associated solver can be used for
rapid prototyping by presenting applications within Static Analysis,
Constraint Satisfactions Problems and Model Checking. In all cases the
complexity result specializes to the worst case time complexity of
classical results.

\bibliographystyle{splncs03}
\bibliography{paper}

\appendix
\newpage
These appendices are not intended for publication and references to
them will be removed in the final version.

\section{Proof of Lemma \ref{lemma:partial-order}}\label{proof:lemma:partial-order}

\begin{proof}

{\bf Reflexivity} $\forall \varrho \in \Delta: \varrho \sqsubseteq \varrho$.

\noindent To show that $\varrho \sqsubseteq \varrho$ let us take $j =
s$. If $rank(R)<j$ then $\varrho(R)=\varrho(R)$ as required. Otherwise
if $rank(R)=j$ and either $R$ is a defined relation or $j=0$, then
form $\varrho(R)=\varrho(R)$ we get $\varrho(R) \subseteq
\varrho(R)$. The last case is when $rank(R)=j$ and $R$ is a
constrained relation. Then from $\varrho(R) = \varrho(R)$ we get
$\varrho(R) \supseteq \varrho(R)$. Thus we get the required $\varrho
\sqsubseteq \varrho$.

\noindent {\bf Transitivity} $\forall \varrho_1, \varrho_2, \varrho_3
\in \Delta: \varrho_1 \sqsubseteq \varrho_2 \wedge \varrho_2
\sqsubseteq \varrho_3 \Rightarrow \varrho_1 \sqsubseteq \varrho_3$.

\noindent Let us assume that $\varrho_1 \sqsubseteq \varrho_2 \wedge
\varrho_2 \sqsubseteq \varrho_3$. From $\varrho_i \sqsubseteq
\varrho_{i+1}$ we have $j_i$ such that conditions
(\ref{itm:rank-less})--(\ref{itm:rank-s}) are fulfilled for
$i=1,2$. Let us take $j$ to be the minimum of $j_1$ and $j_2$. Now we
need to verify that conditions
(\ref{itm:rank-less})--(\ref{itm:rank-s}) hold for $j$. If $rank(R)<j$
we have $\varrho_1(R) = \varrho_2(R)$ and $\varrho_2(R) =
\varrho_3(R)$. It follows that $\varrho_1(R) = \varrho_3(R)$, hence
(\ref{itm:rank-less}) holds. Now let us assume that $rank(R)=j$ and
either $R$ is a defined relation or $j=0$. We have $\varrho_1(R)
\subseteq \varrho_2(R)$ and $\varrho_2(R) \subseteq \varrho_3(R)$ and
from transitivity of $\subseteq$ we get $\varrho_1(R) \subseteq
\varrho_3(R)$, which gives (\ref{itm:rank-eq-def}). Alternatively
$rank(R)=j$ and $R$ is a constrained relation. We have $\varrho_1(R)
\supseteq \varrho_2(R)$ and $\varrho_2(R) \supseteq \varrho_3(R)$ and
from transitivity of $\supseteq$ we get $\varrho_1(R) \supseteq
\varrho_3(R)$, thus (\ref{itm:rank-eq-con}) holds. Let us now assume
that $j \neq s$, hence $\varrho_i(R) \neq \varrho_{i+1}(R)$ for some
$R \in \R$ and $i=1,2$. Without loss of generality let us assume that
$\varrho_1(R) \neq \varrho_2(R)$. In case $R$ is a defined relation we
have $\varrho_1(R) \subsetneq \varrho_2(R)$ and $\varrho_2(R)
\subseteq \varrho_3(R)$, hence $\varrho_1(R) \neq
\varrho_3(R)$. Similarly in case $R$ is a constrained relation we have
$\varrho_1(R) \supsetneq \varrho_2(R)$ and $\varrho_2(R) \supseteq
\varrho_3(R)$. Hence $\varrho_1(R) \neq \varrho_3(R)$, and
(\ref{itm:rank-s}) holds.

\noindent {\bf Anti-symmetry} $\forall \varrho_1, \varrho_2 \in \Delta: \varrho_1
\sqsubseteq \varrho_2 \wedge \varrho_2 \sqsubseteq \varrho_1
\Rightarrow \varrho_1 = \varrho_2$.

\noindent Let us assume $\varrho_1 \sqsubseteq \varrho_2$ and
$\varrho_2 \sqsubseteq \varrho_1$. Let $j$ be minimal such that
$rank(R)=j$ and $\varrho_1(R) \neq \varrho_2(R)$ for some $R \in
\R$. If $j=0$ or $R$ is a defined relation, then we have $\varrho_1(R)
\subseteq \varrho_2(R)$ and $\varrho_2(R) \subseteq \varrho_1(R)$.
Hence $\varrho_1(R) = \varrho_2(R)$ which is a
contradiction. Similarly if $R$ is a constrained relation we have
$\varrho_1(R) \supseteq \varrho_2(R)$ and $\varrho_2(R) \supseteq
\varrho_1(R)$. It follows that $\varrho_1(R) = \varrho_2(R)$, which
again is a contradiction. Thus it must be the case that $\varrho_1(R)
= \varrho_2(R)$ for all $R \in \R$.
\qed
\end{proof}

\section{Proof of Lemma \ref{lemma:complete-lattice}}\label{proof:lemma:complete-lattice}

\begin{proof}
  First we prove that $\bigsqcap M$ is a lower bound of $M$; that is
  $\bigsqcap M \sqsubseteq \varrho$ for all $\varrho \in M$. Let $j$
  be maximum such that $\varrho \in M_j$; since $M=M_0$ and $M_j
  \supseteq M_{j+1}$ clearly such $j$ exists. From definition of $M_j$
  it follows that $(\bigsqcap M)(R)=\varrho(R)$ for all $R$ with
  $rank(R)<j$; hence (\ref{itm:rank-less}) holds.

  \noindent If $rank(R)=j$ and either $R$ is a defined relation or
  $j=0$ we have $(\bigsqcap M)(R)= \bigcap \{ \varrho'(R) \mid
  \varrho' \in M_j \} \subseteq \varrho(R)$ showing that
  (\ref{itm:rank-eq-def}) holds.

  \noindent Similarly, if $R$ is a constrained relation with
  $rank(R)=j$ we have $(\bigsqcap M)(R)= \bigcup \{ \varrho'(R) \mid
  \varrho' \in M_j \} \supseteq \varrho(R)$ showing that
  (\ref{itm:rank-eq-con}) holds.

  \noindent Finally let us assume that $j \neq s$; we need to show
  that there is some $R$ with $rank(R)=j$ such that $(\bigsqcap
  M)(R)\neq\varrho(R)$. Since we know that $j$ is maximum such that
  $\varrho \in M_j$, it follows that $\varrho \notin M_{j+1}$, hence
  there is a relation $R$ with $rank(R)=j$ such that $(\bigsqcap
  M)(R)\neq\varrho(R)$; thus (\ref{itm:rank-s}) holds.

  \noindent Now we need to show that $\bigsqcap M$ is the greatest
  lower bound. Let us assume that $\varrho' \sqsubseteq \varrho$ for
  all $\varrho \in M$, and let us show that $\varrho' \sqsubseteq
  \bigsqcap M$. If $\varrho' = \bigsqcap M$ the result holds vacuously,
  hence let us assume $\varrho' \neq \bigsqcap M$. Then there exists a
  minimal $j$ such that $(\bigsqcap M)(R)\neq\varrho'(R)$ for some $R$
  with $rank(R)=j$. Let us first consider $R$ such that
  $rank(R)<j$. By our choice of $j$ we have $(\bigsqcap M)(R) =
  \varrho'(R)$ hence (\ref{itm:rank-less}) holds.

  \noindent Next assume that $rank(R)=j$ and either $R$ is a defined
  relation of $j=0$. Then $\varrho' \sqsubseteq \varrho$ for all
  $\varrho \in M_j$. It follows that $\varrho'(R) \subseteq
  \varrho(R)$ for all $\varrho \in M_j$. Thus we have $\varrho'(R)
  \subseteq \bigcap \{ \varrho(R) \mid \varrho \in M_j \}$. Since
  $(\bigsqcap M)(R) = \bigcap \{ \varrho(R) \mid \varrho \in M_j \}$,
  we have $\varrho'(R) \subseteq (\bigsqcap M)(R)$ which proves
  (\ref{itm:rank-eq-def}).

  \noindent Now assume $rank(R)=j$ and $R$ is a constrained
  relation. We have that $\varrho' \sqsubseteq \varrho$ for all
  $\varrho \in M_j$. Since $R$ is a constrained relation it follows
  that $\varrho'(R) \supseteq \varrho(R)$ for all $\varrho \in
  M_j$. Thus we have $\varrho'(R) \supseteq \bigcup \{ \varrho(R) \mid
  \varrho \in M_j \}$. Since $(\bigsqcap M)(R) = \bigcup \{ \varrho(R)
  \mid \varrho \in M_j \}$, we have $\varrho'(R) \supseteq (\bigsqcap
  M)(R)$ which proves (\ref{itm:rank-eq-con}).

  \noindent Finally since we assumed that $(\bigsqcap
  M)(R)\neq\varrho'(R)$ for some $R$ with $rank(R)=j$, it follows that
  (\ref{itm:rank-s}) holds. Thus we proved that $\varrho' \sqsubseteq
  \bigsqcap M$.
  \qed

\end{proof}

\section{Proof of Proposition \ref{prop:moore-family}}\label{proof:prop:moore-family}

In order to prove Proposition \ref{prop:moore-family} we first state
and prove two auxiliary lemmas.

\begin{definition}\label{def:order-j}
\noindent We introduce  an ordering $\subseteq_{/j}$ defined by $\varrho_1
\subseteq_{/j} \varrho_2$ if and only if
\begin{itemize}
\item $\forall R: rank(R) < j \Rightarrow \varrho_1(R) = \varrho_2(R)$
\item $\forall R: rank(R) = j \Rightarrow \varrho_1(R) \subseteq
  \varrho_2(R)$
\end{itemize}
\end{definition}

\begin{lemma}\label{lemma:cond}
  Assume a condition $\Cond$ occurs in $cl_j$, and let $\varsigma$ be
  a valuation of free variables in $\Cond$. If $\varrho_1
  \subseteq_{/j} \varrho_2$ and $(\varrho_1,\varsigma) \models \Cond$
  then $(\varrho_2,\varsigma) \models \Cond$.
\end{lemma}
\begin{proof}
  We proceed by induction on $j$ and in each case perform a structural
  induction on the form of the condition $\Cond$ occurring in
  $cl_j$.\newline \textbf{Case: }$\Cond=R(\vec x)$ \newline Assume
  $\varrho_1 \subseteq_{/j} \varrho_2$ and
$$
(\varrho_1,\varsigma) \models R(\vec x)
$$
From Table \ref{LFPSemantics} it follows that
$$
\sem{\vec x}([\,], \varsigma) \in \varrho_1(R)
$$
Depending of the rank of $R$ we have two sub-cases.\newline (1) Let
$rank(R)<j$, then from Definition \ref{def:order-j} we know that
$\varrho_1(R) = \varrho_2(R)$ and hence
$$
\sem{\vec x}([\,], \varsigma) \in \varrho_2(R)
$$
Which according to Table \ref{LFPSemantics} is equivalent to
$$
(\varrho_2,\varsigma) \models R(\vec x)
$$
(2) Let us now assume $rank(R)=j$, then from Definition
\ref{def:order-j} we know that $\varrho_1(R) \subseteq \varrho_2(R)$ and hence
$$
\sem{\vec x}([\,], \varsigma) \in \varrho_2(R)
$$
which is equivalent to
$$
(\varrho_2,\varsigma) \models R(\vec x)
$$
and finishes the case. \newline \textbf{Case: }$\Cond=\neg R(\vec x)$
\newline Assume $\varrho_1 \subseteq_{/j} \varrho_2$ and
$$
(\varrho_1,\varsigma) \models \neg R(\vec x)
$$
From Table \ref{LFPSemantics} it follows that
$$
\sem{\vec x}([\,], \varsigma) \notin \varrho_1(R)
$$
Since $rank(R)<j$, then from Definition \ref{def:order-j} we have
$\varrho_1(R) = \varrho_2(R)$ and hence
$$
\sem{\vec x}([\,], \varsigma) \notin \varrho_2(R)
$$
Which according to Table \ref{LFPSemantics} is equivalent to
$$
(\varrho_2,\varsigma) \models \neg R(\vec x)
$$

\noindent \textbf{Case: }$\Cond=\Cond_1 \wedge \Cond_2$

\noindent Assume $\varrho_1 \subseteq_{/j} \varrho_2$ and
$$
(\varrho_1,\varsigma) \models \Cond_1 \wedge \Cond_2
$$
From Table \ref{LFPSemantics} it follows that
$$
(\varrho_1,\varsigma) \models \Cond_1 \text{ and } (\varrho_1,\varsigma) \models \Cond_2
$$
The induction hypothesis gives
$$
(\varrho_2,\varsigma) \models \Cond_1 \text{ and } (\varrho_2,\varsigma) \models \Cond_2
$$
Hence we have 
$$
(\varrho_2,\varsigma) \models \Cond_1 \wedge \Cond_2
$$

\noindent \textbf{Case: }$\Cond=\Cond_1 \vee \Cond_2$

\noindent Assume $\varrho_1 \subseteq_{/j} \varrho_2$ and
$$
(\varrho_1,\varsigma) \models \Cond_1 \vee \Cond_2
$$
From Table \ref{LFPSemantics} it follows that
$$
(\varrho_1,\varsigma) \models \Cond_1 \text{ or } (\varrho_1,\varsigma) \models \Cond_2
$$
The induction hypothesis gives
$$
(\varrho_2,\varsigma) \models \Cond_1 \text{ or } (\varrho_2,\varsigma) \models \Cond_2
$$
Hence we have 
$$
(\varrho_2,\varsigma) \models \Cond_1 \vee \Cond_2
$$

\noindent \textbf{Case: }$\Cond=\exists x: \CondPr$

\noindent Assume $\varrho_1 \subseteq_{/j} \varrho_2$ and
$$
(\varrho_1,\varsigma) \models \exists x: \CondPr
$$
From Table \ref{LFPSemantics} it follows that
$$
\exists a \in \U: (\varrho_1,\varsigma[x \mapsto a]) \models \CondPr
$$
The induction hypothesis gives
$$
\exists a \in \U: (\varrho_2,\varsigma[x \mapsto a]) \models \CondPr
$$
Hence from Table \ref{LFPSemantics} we have
$$
(\varrho_2,\varsigma) \models \exists x: \CondPr
$$

\noindent \textbf{Case: }$\Cond=\forall x: \CondPr$

\noindent Assume $\varrho_1 \subseteq_{/j} \varrho_2$ and
$$
(\varrho_1,\varsigma) \models \forall x: \CondPr
$$
From Table \ref{LFPSemantics} it follows that
$$
\forall a \in \U: (\varrho_1,\varsigma[x \mapsto a]) \models \CondPr
$$
The induction hypothesis gives
$$
\forall a \in \U: (\varrho_2,\varsigma[x \mapsto a]) \models \CondPr
$$
Hence from Table \ref{LFPSemantics} we have
$$
(\varrho_2,\varsigma) \models \forall x: \CondPr
$$
\qed
\end{proof}

\begin{lemma}\label{lemma:cl}
  If $\varrho=\bigsqcap M$ and $(\varrho', \zeta, \varsigma) \models
  cl_j$ for all $\varrho' \in M$ then $(\varrho,\zeta, \varsigma)
  \models cl_j$.
\end{lemma}
\begin{proof}
We proceed by induction on $j$ and in each case perform a structural
induction on the form of the clause $cl$ occurring in $cl_j$.

\noindent\textbf{Case: }$cl_j=\define(\Cond \Rightarrow R(\vec u))$

\noindent Assume
\begin{equation}
  \forall \varrho' \in M: (\varrho',\zeta,\varsigma) \models \Cond
  \Rightarrow R(\vec u) \label{eq:assumption-def-imply}
\end{equation}
Let us also assume
$$
(\varrho,\varsigma) \models \Cond
$$
Since $\varrho=\bigsqcap M$ we know that
\begin{equation}
\forall \varrho' \in M: \varrho \sqsubseteq \varrho' \label{eq:glb-property}
\end{equation}
Let $R'$ occur in $\Cond$. We have two possibilities; either
$rank(R')=j$ and $R'$ is a defined relation, then from
\eqref{eq:glb-property} if follows that $\varrho(R') \subseteq
\varrho'(R')$. Alternatively $rank(R')<j$ and from
\eqref{eq:glb-property} it follows that $\varrho(R') =
\varrho'(R')$. Hence from Definition \ref{def:order-j} we have that
$\varrho \subseteq_{/j} \varrho'$. Thus from Lemma \ref{lemma:cond} it
follows that
$$
\forall \varrho' \in M: (\varrho',\varsigma) \models \Cond
$$
Hence from \eqref{eq:assumption-def-imply} we have
$$
\forall \varrho' \in M: (\varrho',\zeta,\varsigma) \models R(\vec u)
$$
Which from Table \ref{LFPSemantics} is equivalent to
$$
\forall \varrho' \in M: \sem{\vec u}(\zeta, \varsigma) \in \varrho'(R)
$$
It follows that
$$
\sem{\vec u}(\zeta, \varsigma) \in \bigcup \{ \varrho'(R) \mid \varrho' \in M \} = \varrho(R)
$$
Which from Table \ref{LFPSemantics} is equivalent to
$$
(\varrho,\zeta,\varsigma) \models R(\vec u)
$$
and finishes the case.

\noindent\textbf{Case: }$cl_j=\define(\Def_1 \wedge \Def_2)$

\noindent Assume
$$
  \forall \varrho' \in M: (\varrho',\zeta,\varsigma) \models \Def_1 \wedge \Def_2
$$
From Table \ref{LFPSemantics} we have that for all $\varrho' \in M$
$$
(\varrho',\zeta,\varsigma) \models \Def_1 \textit{ and } (\varrho',\zeta,\varsigma)
\models \Def_2
$$
The induction hypothesis gives
$$
(\varrho,\zeta,\varsigma) \models \Def_1 \textit{ and } (\varrho,\zeta,\varsigma)
\models \Def_2
$$
Hence from Table \ref{LFPSemantics} we have
$$
(\varrho,\zeta,\varsigma) \models \Def_1 \wedge \Def_2
$$

\noindent\textbf{Case: }$cl_j=\define(\forall x: \Def)$

\noindent Assume
\begin{equation}
  \forall \varrho' \in M: (\varrho',\zeta,\varsigma) \models \forall x: \Def
\end{equation}
From Table \ref{LFPSemantics} we have that
$$
\varrho' \in M: \forall a \in \mathcal{U}: (\varrho',\zeta,\varsigma[x
\mapsto a]) \models \Def
$$
Thus
$$
\forall a \in \mathcal{U}: \varrho' \in M: (\varrho',\zeta,\varsigma[x
\mapsto a]) \models \Def
$$
The induction hypothesis gives
$$
\forall a \in \mathcal{U}: (\varrho,\zeta,\varsigma[x
\mapsto a]) \models \Def
$$
Hence from Table \ref{LFPSemantics} we have
$$
 (\varrho,\zeta,\varsigma) \models \forall x: \Def
$$

\noindent\textbf{Case: }$cl_j=\constrain(R(\vec u) \Rightarrow \Cond)$

\noindent Assume
\begin{equation}
  \forall \varrho' \in M: (\varrho',\zeta,\varsigma) \models R(\vec u) 
  \Rightarrow \Cond \label{eq:assumption-con-imply}
\end{equation}
Let us also assume
$$
(\varrho,\zeta,\varsigma) \models R(\vec u)
$$
From Table \ref{LFPSemantics} it follows that
$$
\sem{\vec u}(\zeta, \varsigma) \in \bigcup \{ \varrho'(R) \mid \varrho' \in M \}
$$
Thus there is some $\varrho' \in M$ such that
$$
\sem{\vec u}(\zeta, \varsigma) \in \varrho'(R)
$$
From \eqref{eq:assumption-con-imply} it follows that
$$
(\varrho',\varsigma) \models \Cond
$$
Since $\varrho=\bigsqcap M$ we know that
\begin{equation}
\forall \varrho' \in M: \varrho \sqsubseteq \varrho' \label{eq:glb-property2}
\end{equation}
Let $R'$ occur in $\Cond$. We have two possibilities; either
$rank(R')=j$ and $R'$ is a constrained relation, then from
\eqref{eq:glb-property2} if follows that $\varrho(R') \supseteq
\varrho'(R')$. Alternatively $rank(R')<j$ and from
\eqref{eq:glb-property2} it follows that $\varrho(R') =
\varrho'(R')$. Hence from Definition \ref{def:order-j} we have that
$\varrho' \subseteq_{/j} \varrho$. Thus from Lemma \ref{lemma:cond} it
follows that
$$
(\varrho,\varsigma) \models \Cond
$$
which finishes the case.

\noindent\textbf{Case: }$cl_j=\constrain(\Con_1 \wedge \Con_2)$

\noindent Assume
$$
  \forall \varrho' \in M: (\varrho',\zeta,\varsigma) \models \Con_1 \wedge \Con_2
$$
From Table \ref{LFPSemantics} we have that for all $\varrho' \in M$
$$
(\varrho',\zeta,\varsigma) \models \Con_1 \textit{ and } (\varrho',\zeta,\varsigma)
\models \Con_2
$$
The induction hypothesis gives
$$
(\varrho,\zeta,\varsigma) \models \Con_1 \textit{ and } (\varrho,\zeta,\varsigma)
\models \Con_2
$$
Hence from Table \ref{LFPSemantics} we have
$$
(\varrho,\zeta,\varsigma) \models \Con_1 \wedge \Con_2
$$

\noindent\textbf{Case: }$cl_j=\constrain(\forall x: \Con)$

\noindent Assume
\begin{equation}
  \forall \varrho' \in M: (\varrho',\zeta,\varsigma) \models \forall x: \Con
\end{equation}
From Table \ref{LFPSemantics} we have that
$$
\varrho' \in M: \forall a \in \mathcal{U}: (\varrho',\zeta,\varsigma[x
\mapsto a]) \models \Con
$$
Thus
$$
\forall a \in \mathcal{U}: \varrho' \in M: (\varrho',\zeta,\varsigma[x
\mapsto a]) \models \Con
$$
The induction hypothesis gives
$$
\forall a \in \mathcal{U}: (\varrho,\zeta,\varsigma[x
\mapsto a]) \models \Con
$$
Hence from Table \ref{LFPSemantics} we have
$$
 (\varrho,\zeta,\varsigma) \models \forall x: \Con
$$
\qed
\end{proof}

\noindent{\bf Proposition \ref{prop:moore-family}}: Assume $cls$ is a
stratified LFP formula, $\varsigma_0$ and $\zeta_0$ are
interpretations of the free variables and function symbols in $cls$,
respectively. Furthermore, $\varrho_0$ is an interpretation of all
relations of rank 0. Then $\{ \varrho \mid (\varrho, \zeta_0,
\varsigma_0) \models cls \wedge \forall R: rank(R) = 0 \Rightarrow
\varrho(R) \supseteq \varrho_0(R) \}$ is a Moore family.
\begin{proof}
The result follows from Lemma \ref{lemma:cl}.
\qed
\end{proof}

\section{Proof of Proposition
  \ref{prop:complexity}}\label{proof:prop:complexity}
{\bf Proposition \ref{prop:complexity}}:
For a finite universe $\U$, the best solution $\varrho$ such that
  $\varrho_0 \sqsubseteq \varrho$ of a LFP formula $cl_1, \ldots,
  cl_s$ (w.r.t. an interpretation of the constant symbols) can be
  computed in time
\[
\mathcal{O}(|\varrho_0| + \sum_{1\leq i \leq s} |cl_i||\mathcal{U}|^{k_i})
\]
where $k_i$ is the maximal nesting depth of quantifiers in the $cl_i$
and $|\varrho_0|$ is the sum of cardinalities of predicates
$\varrho_0(R)$ of rank $0$. We also assume unit time hash table
operations (as in \cite{bib:complex}).
\begin{proof}
  Let $cl_i$ be a clause corresponding to the i-th layer. Since $cl_i$
  can be either a define clause, or a constrain clause, we have two
  cases.

  Let us first assume that $cl_i=\define(\Def)$; the proof proceed in
  three phases. First we transform $\Def$ to $\DefPr$ by replacing
  every universal quantification $\forall x: \Def_{\it cl}$ by the conjunction
  of all $|\U|$ possible instantiations of $\Def_{\it cl}$, every existential
  quantification $\exists x: \Cond$ by the disjunction of all $|\U|$
  possible instantiations of $\Cond$ and every universal quantification
  $\forall x: \Cond$ by the conjunction of all $|\U|$ possible
  instantiations of $\Cond$. % The formal definition of the
  % transformation is given in Figure XX.
  The resulting clause $\DefPr$
  is logically equivalent to $\Def$ and has size
  \begin{equation}
    \O(|\U|^{k}|\Def|) \label{eq:def-size}
  \end{equation}
  where $k$ is the maximal nesting depth of quantifiers in $\Def$.
  Furthermore, $\DefPr$ is {\it boolean}, which means that there are
  no variables or quantifiers and all literals are viewed as nullary
  predicates.

% \begin{figure}
% $$
% \begin{array}{lll}
%   f(\Cond \Rightarrow R(\vec u)) & = & f(\Cond) \Rightarrow R(\vec u) \\
%   f(\forall x: \Def) & = & \bigwedge_{a\in\U}\Def[a/x] \\
%   f(\Def_1 \wedge \Def_2) & = & f(\Def_1) \wedge (\Def_2) \\\\
%   f(\exists x: \Cond) & = & \bigvee_{a\in\U}\Cond[a/x] \\
%   f(\forall x: \Cond) & = & \bigwedge_{a\in\U}\Cond[a/x] \\
% f(\Cond_1 \wedge \Cond_2) & = & f(\Cond_1) \wedge (\Cond_2) \\
% f(\Cond_1 \vee \Cond_2) & = & f(\Cond_1) \vee (\Cond_2) \\
% f(R(\vec x)) & = & R(\vec x) \\
% f(\neg R(\vec x)) & = & \neg R(\vec x) \\
% f(\true) & = & \true \\
% f(\false) & = & \false \\
%   f(\Cond) & = & \Cond
% \end{array}
% $$
% \caption{First phase}
% \label{fig:phase1}
% \end{figure}

  In the second phase we transform the formula $\DefPr$, being the
  result of the first phase, into a sequence of formulas
  $\DefDPr=\DefPr_1, \ldots, \DefPr_l$ as follows. We first replace
  all top-level conjunctions in $\DefPr$ with ",". Then we
  successively replace each formula by a sequence of simpler ones
  using the following rewrite rule
\[\Cond_1 \vee \Cond_2
  \Rightarrow R(\vec u) \mapsto \Cond_1 \Rightarrow Q_{new},
  \Cond_2 \Rightarrow Q_{new}, Q_{new} \Rightarrow R(\vec u)
\] 
where
  $Q_{new}$ is a fresh nullary predicate that is generated for each
  application of the rule. The transformation is completed as soon as no
  replacement can be done. The conjunction of the resulting define
  clauses is logically equivalent to $\DefPr$.

  To show that this process terminates and that the size of $\DefDPr$ is
  at most a constant times the size of the input formula $\DefPr$ , we
  assign a cost to the formulae.  Let us define the cost of a sequence
  of clauses as the sum of costs of all occurrences of predicate
  symbols and operators (excluding ","). In general, the cost of
  a symbol or operator is 1 except disjunction that counts 6. Then the
  above rule decreases the cost from $k + 7$ to $k + 6$, for suitable
  value of k. Since the cost of the initial sequence is at most 6
  times the size of $\Def$, only a linear number of rewrite steps can
  be performed. Since each step increases the size at most by a
  constant, we conclude that the $\DefDPr$ has increased just by a
  constant factor. Consequently, when applying this transformation to
  $\DefPr$, we obtain a boolean formula without sharing of size as in
  \eqref{eq:def-size}.  

The third phase solves the system that is a
  result of phase two, which can be done in linear time by the
  classical techniques of e.g. \cite{bib:hornsat}.

Let us now assume that the $cl_i=\constrain(\Con)$. We begin by
transforming $\Con$ into a logically equivalent (modulo fresh
predicates) {\it define} clause. The transformation is done by
function $f_i$ defined as
% $$
% \begin{array}{lll}
% g_i(\constrain(\Con)) & = & let\ (\Def, \DefPr) = f_i(\Con)\ in \\
% & & \define(\Def),define(\DefPr) \\
% f_i (\forall x: \Con) & = & let\ (\Def, \DefPr) = f_i(\Con)\ in \\
% & & (\forall x: \Def, \forall x: \DefPr) \\
% f_i(\Con_1 \wedge \Con_2) & = & let\ (\Def_1, \Def_1') = f_i(\Con_1)\
% in \\
% & & let\ (\Def_2, \Def_2') = f_i(\Con_2)\ in \\
% & & (\Def_1 \wedge \Def_2, \Def_1' \wedge \Def_2') \\
% f_i (R(\vec u) \Rightarrow \Cond) & = & let\ \CondPr = \neg
% \Cond[R^\complement(\vec u) / \neg R(\vec u)]\ in \\
% & & let\ \Cond'' = \Cond[true / (R'(\vec v) \mid rank(R')=i)]\ in \\
% & & (\CondPr \Rightarrow R^\complement(\vec u), \Cond'' \wedge \neg
% R^\complement(\vec u) \Rightarrow R(\vec u)) \\
% \end{array}
% $$
$$
\begin{array}{lll}
  f_i(\constrain(\Con)) & = & \define(g(\Con)),\define(h_i(\Con)) \\ \\
  g (\forall x: \Con) & = & \forall x: g(\Con) \\
  g(\Con_1 \wedge \Con_2) & = & g(\Con_1) \wedge g(\Con_2) \\
  g (R(\vec u) \Rightarrow \Cond) & = & (\neg \Cond[R^\complement(\vec
  u) / \neg R(\vec u)] \Rightarrow R^\complement(\vec u)) \\ \\

  h_i(\forall x: \Con) & = & \forall x: h_i(\Con) \\
  h_i(\Con_1 \wedge \Con_2) & = & h_i(\Con_1) \wedge h_i(\Con_2) \\
  h_i (R(\vec u) \Rightarrow \Cond) & = & let\ \CondPr = \Cond[true / (R'(\vec v) \mid rank(R')=i)]\ in \\
  & & \CondPr \wedge \neg R^\complement(\vec u) \Rightarrow R(\vec u) \\
\end{array}
$$
where $R^\complement$ is a new predicate corresponding to the
complement of $R$. The size of the formula increases by a number of
constraint predicates; hence the size of the input formula is increased
by a constant factor. Then the proof proceeds as in case of {\it
  define} clause.

The three phases of the transformation result in the sequence of
define clauses of size
\[
\O(\sum_{1\leq i \leq s} |cl_i||\mathcal{U}|^{k_i})
\]
which can then be solved in linear time. We also need to take into
account the size of the initial knowledge i.e.~the cardinality of all
predicates of rank $0$; thus the overall worst case complexity is
\[
\O(|\varrho_0| + \sum_{1\leq i \leq s} |cl_i||\mathcal{U}|^{k_i})
\]
\qed
\end{proof}

\end{document}